\documentclass[twocolumn,preprintnumbers,amsmath,amssymb,
superscriptaddress,nofootinbib]{revtex4}
\usepackage[paperwidth=210mm,paperheight=297mm,centering,hmargin=2cm,vmargin=2.5cm]{geometry}

\usepackage{graphicx}
\usepackage{bbm}
\usepackage{url, xcolor}
\usepackage{amsmath}
\usepackage{amsthm}
\usepackage{amssymb}
\usepackage{array}

\usepackage[latin10]{inputenc}

\usepackage{amsfonts}
\usepackage{amssymb}
\usepackage[english]{babel}

\usepackage[font={small,it}]{caption}

\newcommand{\bra}[1]{\langle #1|}
\newcommand{\ket}[1]{|#1\rangle}
\newcommand{\braket}[2]{\langle #1|#2\rangle}
\newcommand{\ketbra}[2]{|#1\rangle\!\langle#2|}

\newtheorem{thm}{Theorem}
\newtheorem{cor}[thm]{Corollary}

\DeclareMathOperator{\Tr}{Tr}

\let\oldmarginpar\marginpar
\renewcommand\marginpar[1]{\-\oldmarginpar[\raggedleft\marginparsize #1]%
{\raggedright\marginparsize #1}}

\usepackage{amsmath,amssymb,amsthm,bbold,braket,color,fancyvrb,graphicx,relsize}

\newcommand{\cev}[1]{\reflectbox{\ensuremath{\vec{\reflectbox{\ensuremath{#1}}}}}}
\newcommand{\minus}{\scalebox{0.75}[1.0]{$-$}}

\DeclareMathOperator{\erf}{erf}

\newtheorem{thmdef}{Definition}

\begin{document}

\setlength{\tabcolsep}{1ex}

\title{Tunnelling necessitates negative Wigner function}

\author{Yin Long Lin}
\affiliation{Blackett Laboratory, Imperial College London, SW7 2BB}
\author{Oscar C. O. Dahlsten}
\affiliation{Clarendon Laboratory, University of Oxford, Parks Road, Oxford OX13PU, United Kingdom}
\affiliation{London Institute for Mathematical Sciences,
South Street 35a, Mayfair, London, W1K 2XF}
%
%
\date{\today}

\begin{abstract} 
We consider in what sense quantum tunnelling is associated with non-classical probabilistic behaviour. We use the Wigner function quasi-probability description of quantum states. We give a definition of tunnelling that allows us to say whether in a given scenario there is tunnelling or not. We prove that this can only happen if either the Wigner function is negative and/or a certain measurement operator which we call the tunnelling rate operator has a negative Wigner function. 
\end{abstract}

\maketitle

\setcounter{secnumdepth}{0}
\section{Introduction}
Quantum systems can `tunnel' through barriers, a wave-like behaviour that is a prototypical of non-classical behaviour. Tunnelling plays a crucial role in a range of systems, including the theory of radioactivity, in which tunnelling is used to explain how radioactivity emerges from a nucleus despite the surrounding Coulomb barrier \cite{QTR}; and more recent applications include adiabatic quantum computing and annealing \cite{IAM}, and quantum electronics in nano-devices \cite{IA}.

We here try to clarify our understanding of in what sense tunnelling is non-classical, in particular to what extent it is associated with non-classical probabilistic behaviour. This is in order to understand tunnelling better, and we hope that this understanding will help in clarifying to what extent tunnelling allows for non-classical information processing. For example, the crucial component of information processing is the mathematical structure of the probabilistic system, rather than the physical parameters involved.

We find that tunnelling is indeed associated with non-classical probabilistic behaviour. In particular, we show that the Wigner function of the state, a commonly-used and powerful phase space quasi-probability representation of quantum states, and/or a certain operator, which we termed in this Article as tunnelling rate operator, has to contain negative values at some phase space points to be able to tunnel. Here we list two examples: 

\begin{enumerate}
	\item A Gaussian wave packet inbound on a square potential has a positive Wigner function \cite{PSH}. However, once it hits the potential, it loses its Gaussian nature and this forces it to have a negative Wigner function while part of it tunnels;
	
	\item A ground state of a simple harmonic oscillator, which has a positive Wigner function, has some probability of being found 'inside' the binding potential region. The tunnelling rate operator we will define in this Article contains negative values in this case. 
\end{enumerate} 

To be able to derive our main statement we needed to define clearly what tunnelling is and this definition is also one of the contributions of this Article. Finally we note that Wigner function as a real-vector representation of quantum states, fits into the generalised probabilistic framework, and therefore allows for extension to post-quantum theories. To our knowledge, this has not been realised before. Lastly, we discuss the possibility of studying tunnelling in post-quantum theories.  

We proceed as follows. Firstly we give a brief technical introduction to tunnelling and Wigner functions. Then we give new results, our definition of tunnelling and our main theorem. We discuss the interpretation of the theorem. Finally we show that phase space quantum mechanics could be fitted into the generalised probabilistic framework, and demonstrate possible road maps into studying of the phenomenon of tunnelling in post-quantum theories.  


We now give a technical introduction. \textbf{Quantum tunnelling} refers to the phenomenon where a quantum system with insufficient energy penetrates and passes through a potential barrier, which defies the laws of classical mechanics \cite{QTR}. It is usually demonstrated mathematically by solving the one spatial-dimension Schr\"{o}dinger equation with a rectangular potential barrier, such that tunnelling \cite{QMC}\cite{QMG}\cite{QMS}:

\begin{thmdef}
\label{thm:SD}
\textbf{Standard definition of tunnelling.} For a system in a rectangular potential barrier with the form 
\begin{equation}
\label{eq:RPB}
 V(x) = \begin{cases}
  V_0 & x \geq 0,\\
  0 & \text{ otherwise,} 
 \end{cases}
\end{equation}
an energy eigenstate with definite energy $E$ is a tunnelling state iff $E < V_0$ and the probability of finding the state in the region $x \geq 0$ is non-zero.
\end{thmdef} 

Similar behaviours had been studied substantially in quantum systems in other potentials, such as series of rectangular barriers \cite{QTRK} and double wells \cite{QTJM}, where there is finite probability of locating an energy eigenstate in classically forbidden region. However, it is difficult to conceive an analogous definition for quantum systems in a superposition of energy eigenstates or in a more complicated potential, because it is less clear of what the corresponding classically forbidden regions are for both cases.

In order to provide a more general definition of quantum tunnelling and to understand this phenomenon in terms of non-classical probabilistic behaviour, we have chosen the \textbf{phase space formulation of quantum mechanics} as the fundamental framework. In this formulation, the state of a quantum system is described by a quasi-probability distribution, and observables are replaced by ordinary c-number functions in phase space. Mathematically, a quantum state described by a vector in the Hilbert space formulation $\ket{\psi}$ can be transformed to a real function \textit{Wigner function} $W(x,p)$ in the phase space formulation as
\begin{equation}
\label{eq:wig-pure}
 W(x,p) = \frac{1}{\pi \hbar} \int e^{2ipy/\hbar} \psi^*(x+y) \psi(x-y) \,dy.
\end{equation}
Such function satisfies the normalisation condition $\int W(x,p) \,dx\,dp = 1$. However, a Wigner function, in general, is not everywhere positive, and therefore cannot be considered as a legitimate joint probability distribution in phase space. It was demonstrated that a Wigner function of a pure continuous-variable state is positive if and only if the state is Gaussian, which is known as \textit{Hudson's theorem} \cite{PSH}.

\begin{figure}[!htb]
 \centering
  \includegraphics[scale=0.4]{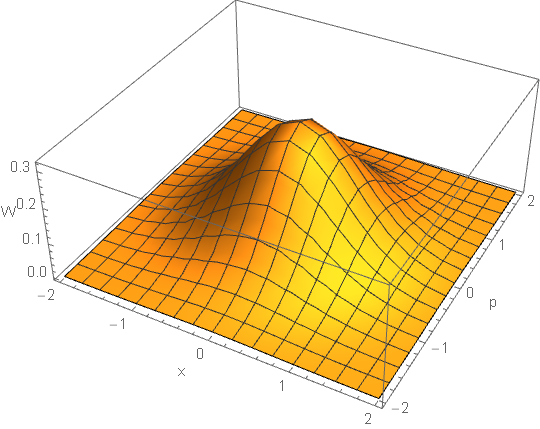}
  \includegraphics[scale=0.4]{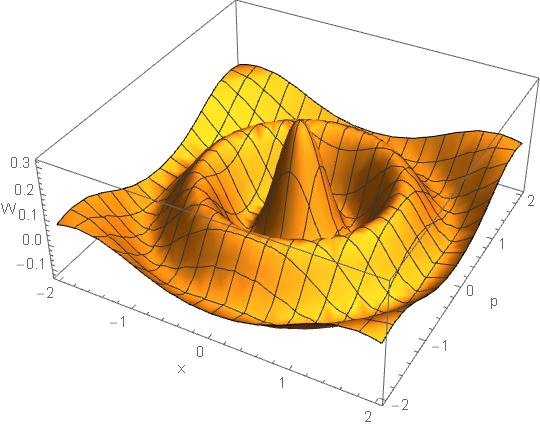}
  \caption{Examples of Wigner functions. The top one is positive everywhere and Gaussian. The lower one has negative values in places and is non-Gaussian. }
\label{fig:GPSG}
\end{figure}

A more general approach to describe states and operators $\hat{\Omega}$ as phase space functions $\mathcal{O}(x,p)$ is via \textit{Weyl transformation}:
\begin{equation}
\label{eq:weyl}
 \mathcal{O}(x,p) = \frac{1}{\pi \hbar} \int e^{2ipy/\hbar} \braket{x+y|\hat{\Omega}|x-y} \,dy,
\end{equation}
and the inverse transformation, \textit{Wigner transformation}, is given by the formula:
\begin{multline}
\label{eq:wigner}
 \hat{\Omega} = \frac{1}{(2\pi)^2} \iiiint \mathcal{O}(x,p) e^{i[\alpha(\hat{X}-x) + \beta(\hat{P} - p)]}\,d\alpha\,d\beta\,dx\,dp.
\end{multline}
It is clear, from \eqref{eq:wig-pure} and \eqref{eq:weyl}, the Wigner function of a pure state is the Weyl transformation of a pure density operator $\ket{\psi} \bra{\psi}$. With the fact that a general density operator represented by $\hat{\rho} = \sum_i \lambda_i \ket{\psi_i} \bra{\psi_i}$, any mixed state in the phase space representation is described by a convex sum of pure Wigner functions $W(x,p) = \sum_i \lambda_i W_i(x,p)$, where $W_i(x,p)$ are the pure Wigner functions and $\sum_i \lambda_i = 1$. 

Another point to note about the transition between Hilbert space formulation to phase space formulation of quantum mechanics is the preservation of non-commutative nature of operators via the introduction of \textit{star product} $\star$,
\begin{equation}
\label{eq:star}
 \star \equiv \exp \left[ \frac{i\hbar}{2} (\vec{\partial}_x \cev{\partial}_p - \vec{\partial}_p \cev{\partial}_x) \right],
\end{equation}
such that the mapping can be written as
\begin{equation}
 \hat{\Omega}_1 \hat{\Omega}_2 \rightarrow \mathcal{O}_1 (x,p) \star \mathcal{O}_2 (x,p),
\end{equation}
where $\mathcal{O}_i(x,p)$ is the Weyl transformation of the corresponding operator $\hat{\Omega}_i$.

In the phase space formulation, the probability of a measurement outcome $P(\Omega = \omega)$ corresponding to $\Tr (\ketbra{\omega}{\omega} \ketbra{\psi}{\psi})$ (where $\Omega \ket{\omega} = \omega \ket{\omega}$) for a quantum state $\ket{\psi}$ and a corresponding Wigner function $W_\psi(x,p)$ is 
\begin{equation}
P(\Omega = \omega) = 2\pi\hbar \iint W_{\omega} (x,p) W_{\psi} (x,p)\,dx\,dp,
\end{equation}
where $W_{\omega} (x,p)$ is the Wigner function of the state $\ket{\omega}$. In other words, this probability is simply the integral of the product of Wigner functions corresponding to the measurement outcome and the quantum state itself up to a normalisation constant $2\pi\hbar$. This result is analogous to its classical counterpart described by a distribution function $f(x,p)$, where the probability of a dynamic variable $\Omega = \omega$ is given by 
\begin{equation}
 P(\Omega = \omega) = \iint \delta[\Omega(x, p) - \omega] f(x,p)\,dx\,dp.
,\end{equation}
 It should be clear that the probability corresponding to a measurement of a dynamic variable $\Omega = \omega$ is the inner product of the state of the system $f(x,p)$, and a function that describes the measurement of probability of certain dynamic variable, which is known as an \textit{effect} $\mathcal{E}$. For instance, in the examples above, the effects for calculating the probability $P(\Omega = \omega)$, $\mathcal{E}_{\Omega = \omega}(x, p)$ are $\delta[\Omega(x,p) - \omega]$ and $2 \pi \hbar W_{\omega}(x,p)$ for classical and quantum cases respectively.

From the brief discussions above, it is shown that the mathematical structures of classical and quantum theories in phase space are similar to one another. Such similarities are anticipated by the fact that classical probability theories and finite-dimensional quantum theories can be described by a unified framework known as the \textbf{generalised probabilistic theories} (GPTs) \cite{GPTJ}.

A finite-dimensional probabilistic theory under the framework of GPTs has three major components:
\\
\\
\textbf{1. Preparation.} A state of a finite-dimensional system is represented by a real vector from a \textit{state space}, which is defined by the convex set of extremal states, or \textit{pure states}, and the \textit{null state}, within a finite-dimensional vector space. Such mathematical object completely determines the characteristics of the state of a physical system.
\\
\\
\textbf{2. Transformation}. The evolution of any state vector is represented by a linear map that maps any state vector into another within the state space.
\\
\\
\textbf{3. Measurement}. The probability of a particular measurement outcome for any state is described by the inner product between the state vector and an \textit{effect vector} corresponding to such measurement. The inner product calculated must be ranging between 0 and 1 as a valid probability \cite{GPTH}. The set of valid effects is known as the \textit{effect space}, which is constructed as the convex sum of the set of \textit{pure effects} and \textit{null effect}.

Under the GPT framework, it is possible to construct probabilistic theories other than classical and quantum theories by varying its state space, the effect space and the set of allowed transformations, hence providing a way of generalising existing theories into \textit{post-quantum theories}. It was shown via these constructed post-quantum theories that certain quantum phenomena, such as non-unique decomposition of mixed states into pure states, and no-cloning theorem \cite{GPTB}, which were thought to be novel to quantum systems, were actually generic properties of generalised probabilistic theories. However, these studies are limited to finite-dimensional systems and therefore cannot be immediately applied to the study of quantum tunnelling, which is a wave-mechanical quantum phenomenon. Therefore if it is possible to rewrite infinite-dimensional quantum mechanics in the framework of GPTs, we could determine whether tunnelling is unique to quantum theories, and examine such phenomenon in post-quantum theories.

\section{Results}
\subsection{ Our definition of tunnelling}
 It is clear from the previous section that \text{Def. \ref{thm:SD}} is too restrictive to describe the entire class of tunnelling behaviour.  A general quantitative definition of tunnelling is therefore required to be: 
 
 \begin{enumerate}
 	\item Mathematically well-defined; 
 	\item Able to recover \text{Def. \ref{thm:SD}} as a special case;
 	\item Able to reflect the non-classicality of the behaviour;
 	\item Able to supply a quantitative criteria of tunnelling for systems with general potentials and states without definite energy.
 \end{enumerate}

In order to construct such a definition of tunnelling, we starts with the law of conservation of energy for a classical particle in a potential $V(x)$:

\begin{enumerate}
\item Classical kinetic energy $p^2/2m$ is positive, due to the reality of momentum $p$. By conservation of energy, for a particle with energy $E^*$,
\begin{align}
 E^* - V(x) &= \frac{p^2}{2m} \geq 0 \nonumber \\
 E^* &\geq V(x).
\end{align}
This implies that a particle with energy $E^*$ is not allowed in the region $\{x | V(x)>E^*\}$ classically, where $x$ is a real variable representing position. Denote this region as $\mathcal{X}(E^*)$, as the classically forbidden region for states with energy $E^*$.

\item For two energies $E^*_1$ and $E^*_2$, such that $E^*_1 > E^*_2$, their respective classically forbidden regions satisfy $\mathcal{X}(E^*_1) \subset \mathcal{X}(E^*_2)$, since
\begin{align*}
 \forall x^*: x^* \in \mathcal{X}(E^*_1) &\leftrightarrow x^* \in \{x|V(x) > E^*_1\} \\
 & \leftrightarrow V(x^*) > E^*_1 \\
 & \rightarrow V(x^*) > E^*_2 \\
 & \leftrightarrow x^* \in \{x|V(x) > E^*_2\} \\
 & \leftrightarrow x^* \in \mathcal{X}(E^*_2) \\
 \therefore \mathcal{X}(E^*_1) &\subset \mathcal{X}(E^*_2).
\end{align*}
\item Classically, only a particle with energy $E> E^*$ is allowed to be in $\mathcal{X}(E^*)$, since:

\begin{enumerate}
 \item For $E = E^*$, $\mathcal{X}(E^*)$ is classically forbidden region; 
 \item For $E < E^*$, by (2), $\mathcal{X}(E^*) \subset \mathcal{X}(E)$ is also classically forbidden.
\end{enumerate}
\end{enumerate}

With the concluding statement, and with the general principle that tunnelling is a phenomenon that violates this classical constraint, we formulate the definition of tunnelling to be:

\begin{thmdef}
\label{def:QM-GD}
 \textbf{General Definition of Tunnelling.} For a state in a potential given by $V(x)$, it is tunnelling if and only if there exists some energy $E^*$, such that the probability of locating the state in region where $V(x) > E^*$ is greater than that of measuring the state to have energy $E > E^*$, or mathematically,
 \begin{equation*}
  \exists E^*: P(x | V(x) > E^*) > P(E > E^*).
 \end{equation*}
\end{thmdef}

Notice that this definition does not require the state in question to be quantum or classical, and therefore allows the definition to be applied as a condition of general phenomenon on any physical systems, provided they can be described by energy and position as physical parameters. Another advantage of using such general definition is its capabilities of extending to other non-classical behaviour, such as reflection over barrier as discussed in Appendix A. It is also shown, in Appendix B, that the conventional definition \text{Def. \ref{thm:SD}} could be recovered as a special case of such general definition.

\subsection{Main theorem:\\ Tunnelling necessitates negative Wigner function }
The necessary and sufficient condition for tunnelling, or equivalently the general definition \text{Def. \ref{def:QM-GD}}, are based on cumulative probabilities $P(E > E^*)$ and $P(x|V(x) > E^*)$. In the phase space formulation, they are given by the inner product of the functions representing the difference of the effects corresponding to the two probabilities, and the Wigner function representing the state. The general definition of tunnelling applied in this picture implies:

\begin{thm}
\label{thm:PSTNS}
\textbf{Necessary and Sufficient Conditions of Tunnelling in Phase Space.} A state represented by a distribution function $f(x,p)$ in phase space is tunnelling if and only if there exists some $E^*$:
\begin{multline*}
 \iint \left[\mathcal{E}_{\{x|V(x) > E^*\}}-\mathcal{E}_{E > E^*}\right](x,p) f(x,p)\,dx\,dp > 0,
\end{multline*}
where $\mathcal{E}_{\{x|V(x) > E^*\}}$ and $\mathcal{E}_{E > E^*}$ are the effects corresponding to the probability measurement of outcome $P(x | V(x) > E^*)$ and $P(E > E^*)$ respectively.
\end{thm}

\begin{proof}
This statement is basically a rewritten form of the general definition. As by the definition of the effects $\mathcal{E}$ as stated in the theorem, 
\begin{multline*}
 P(x | V(x) > E^*) = \iint \mathcal{E}_{\{x|V(x) > E^*\}}(x,p) f(x,p) \,dx \,dp; \\
 P(E > E^*) = \iint \mathcal{E}_{E > E^*}(x,p) f(x,p)\,dx\,dp,
\end{multline*}
by \text{Def. \ref{def:QM-GD}}, the necessary and sufficient condition of tunnelling is therefore $\exists E^*$:
\begin{multline*}
 \iint [\mathcal{E}_{\{x|V(x) > E^*\}}
 -\mathcal{E}_{E > E^*}](x,p) f(x,p)\,dx\,dp > 0,
\end{multline*}
which is the theorem to be proved. 
\end{proof}
%


The theorem therefore shows that tunnelling is related to the distribution function and $\mathcal{E}_{E>E^*} - \mathcal{E}_{\{x|V(x)>E^*\}}$, and the latter is denoted as the \textit{tunnelling rate operator} at $E^*$. Two important consequences of \text{Thm. \ref{thm:PSTNS}} are:

\begin{cor}
\label{col:PSTN}
A state is non-tunnelling if, for every energy $E^*$, the tunnelling rate operator, and the distribution function are non-negative everywhere in the phase space.
\end{cor}

\begin{proof}
An equivalent way of expressing the necessary and sufficient conditions of tunnelling is to specify the necessary and sufficient conditions of a non-tunnelling state, which is for all $E^*$: 
\begin{align*}
 \iint \left[\mathcal{E}_{\{x|V(x) > E^*\}}-\mathcal{E}_{E > E^*}\right] f(x,p) \,dx \,dp &\leq 0 \\
 \iint \left[\mathcal{E}_{E > E^*} - \mathcal{E}_{\{x|V(x)>E^*\}}\right] f(x,p) \,dx \,dp &\geq 0.
\end{align*}
It is clear that if both $\mathcal{E}_{E > E^*} - \mathcal{E}_{\{x|V(x)>E^*\}}$ and $f(x,p)$ is non-negative over all of phase space, then this condition will hold. 
\end{proof}

Notice here this theorem only provides a sufficient but not necessary condition for the case of non-tunnelling, as it is not the case where $\mathcal{E}_{E > E^*} - \mathcal{E}_{\{x|V(x) > E^*\}}$ and $f(x,p)$ must be positive for the non-tunnelling condition to be satisfied. Another way of rephrasing the corollary, or indeed as a direct consequence of the main theorem, is:

\begin{cor}
\label{col:PST}
If a state is tunnelling, then either tunnelling rate operator at some $E^*$, or the distribution function, contain negativities.
\end{cor}

\begin{proof}
The result can be directly inferred from \text{Thm. \ref{thm:PSTNS}}, by a simple reshuffling of the inequality that defines tunnelling: for some energy $E^*$,
\begin{align*}
\iint [\mathcal{E}_{\{x|V(x) > E^*\}}-\mathcal{E}_{E > E^*}](x,p) f(x,p)\,dx\,dp > 0, \\
\iint [\mathcal{E}_{E > E^*} -\mathcal{E}_{\{x|V(x) > E^*\}}](x,p) f(x,p)\,dx\,dp < 0.
\end{align*}

Since it is impossible for the statement above to hold if any of the two functions are positive over the entire phase space, then either $\mathcal{E}_{E > E^*} - \mathcal{E}_{\{x|V(x)>E^*\}}$ or $f(x,p)$ must contain negativities.
\end{proof}

The last statement demonstrates the main claim of this paper: under a phase space framework, tunnelling implies negativities in the distribution representing the state, and/or the tunnelling rate operator at some energy levels. In the case of quantum systems, tunnelling implies negativities in the Wigner function of the state and/or the tunnelling rate operator.

\section{Discussion}
\subsection{Quantum vs. Classical case}

The results of the previous section could immediately be applied to the study of generic behaviours of classical and quantum theories. One important result is that classical phase space mechanics do not allow tunnelling, because as shown in Appendix C, both the phase space distributions representing a state and the tunnelling rate operator does not contain non-negativities, and by \text{Cor. \ref{col:PSTN}}, classical states cannot tunnel. However, in quantum theories, the non-negativity conditions of the functions are relaxed. This can be attributed to the following two elements of quantum phase space theory:
\\
\\
\textbf{1. Wigner Function as Quasi-Probability Distribution.} A quantum state in phase space is represented by Wigner function, which generally contains negative values in general. 
 \\
 \\
\textbf{2. Deformation of Effect }$\mathbf{\mathcal{E}_{E > E^*}.}$ Generally speaking, the effect $\mathcal{E}_{\{x | V(x) > E^*\}}$ is identical for both classical and quantum case, while the effect $\mathcal{E}_{E > E^*}$ is altered in the quantum case. Both phenomena can be explained by the alteration in the eigenvalue equations for both position and Hamiltonian operator represented in the phase space picture, as will be shown here.
 
 Generally, for any dynamical variables $\Omega(x,p)$, for both classical and quantum systems, their eigenstate with respect to the corresponding algebra has definite value $\omega$ if one conducts a measurement of $\Omega$ on such system. 
 
 Consider first the dynamical variable $x$ as the position. The eigenvalue equation in classical systems with eigenstate $f_0(x,p)$ and eigenvalue $x_0$ is 
 \begin{equation}
  x f_0(x,p) = x_0 f_0(x,p),
 \end{equation}
 and this equation has an obvious solution of $f_0(x,p) = \delta(x - x_0) g(p)$, where $g$ is an arbitrary continuous positive real function which is normalised to unity $\int g(p)\,dp = 1$. If the state in question is of such a form, the measurement of position on such state must give $x = x_0$. At the same time, such function could be used as an effect to determine what is the probability of a certain state to have a position measurement of $x = x_0$, as the inner product of this delta function against any state would give the proportion of the state with $x = x_0$, i.e. the corresponding probability. Therefore, classically, $\mathcal{E}_{x = x_0} = \delta(x - x_0)$. 
 
 In the case of quantum systems, the product of any operators replaced by star-product of functions as a result of deformation quantisation. Therefore, the eigenvalue problem is mapped to an equation with position eigenstate $W_0(x,p)$ and eigenvalue $x_0$:
 \begin{equation}
  x \star W_0(x,p) = x_0 W_0(x,p).
 \end{equation}
 One simple way of solving the problem is to switch back to the Hilbert space picture, where the equivalent problem is 
 \begin{equation}
  \braket{x|\hat{X}|\psi_0} = x_0 \braket{x|\psi_0},
 \end{equation}
 which has a well-known solution of $\braket{x|\psi_0} = \psi_0(x) = \delta(x-x_0)$. The Wigner function for such a state is given by
 \begin{align}
  &\, \frac{1}{\pi\hbar} \int e^{2ipy/\hbar} \delta(x + y - x_0) \delta(x - y - x_0) \,dy \sim \,\delta(x-x_0)
 \end{align}
 which is again of the form of a Dirac delta function in position space after integrating away the momentum dependence. By the duality of effect and state in quantum theory, the effect for $x = x_0$ is therefore $\mathcal{E}_{x = x_0} = \delta(x-x_0)$. 
 
 Therefore, since both the classical and quantum position eigenstates are of the form of a Dirac delta function $\delta(x-x_0)$, the effect corresponding to measurement in position space are therefore identical, which implies $\mathcal{E}_{\{x | V(x) > E^*\}}$ is of the same form for both classical and quantum calculations.
 
 However, this is not the case for $\mathcal{E}_{E > E^*}$. The energy eigenstate $f_E(x,p)$ with energy $E$ for a classical system satisfies 
 \begin{equation}
  H(x,p) f_E(x,p) = E f_E(x,p),
 \end{equation}
 which gives $f_E(x,p) \sim \delta\left[H(x,p)-E\right]$. The quantum energy eigenvalue equation in phase space for an energy eigenstate $W_E(x,p)$ and energy $E$, however, is altered to be
 \begin{equation}
  H(x,p) \star W_E(x,p) = E W_E(x,p).
 \end{equation}
 Here, the energy eigenstate as a Weyl map of $\hat{\rho} = \ket{E}\bra{E}$ is no longer of the form $\delta \left[H(x,p) - E\right]$, as the introduction of star product implies the functional dependence of energy eigenstate is no longer purely on the functional form of the dynamical variable. This can be seen by the Bopp shift representation of star product \eqref{eq:star}, where the functional dependence of the energy eigenstate depends also on the position and momentum derivatives of the Hamiltonian function. Therefore, the energy eigenstate $\mathcal{E}_{E > E^*}$ is not identical to its classical counterpart. 


 With the energy effect $\mathcal{E}_{E > E^*}(x,p) = 2\pi\hbar \int_{E^*}^\infty W_{E'}(x,p)\,dE'$, and the fact that Wigner functions are generally not completely positive, $\mathcal{E}_{E > E^*}(x,p)$ generally contains negativity. Therefore, the function $\mathcal{E}_{E > E^*} - \mathcal{E}_{\{x | V(x) > E^*\}}$ generally contains negativity, given that $\mathcal{E}_{\{x | V(x) > E^*\}}$ is non-negative. This demonstrates that how this operator can violate the classical case by containing negativities.

\subsection{Tunnelling in pure Gaussian states} A special class of states in quantum mechanics is Gaussian states, which has non-negative Wigner functions in phase space representation. By \text{Cor. \ref{col:PSTN}}, a statement for Gaussian states could be made as:

\begin{cor}
\label{cor:PTG}
 A Gaussian state can tunnel only if the function $\mathcal{E}_{E > E^*} - \mathcal{E}_{\{x | V(x) > E^*\}}$ contains negativities for some $E^*$. 
\end{cor}

\begin{proof}
By \text{Cor. \ref{col:PSTN}}, if a state is tunnelling, then either $\mathcal{E}_{E > E^*} - \mathcal{E}_{\{x | V(x) > E^*\}}$ or $W(x,p)$ contains negativities for some energy $E^*$. By Hudson's theorem \cite{PSH}, the Wigner function representing Gaussian state is non-negative over all of phase space. Therefore, if a Gaussian state is tunnelling, then $\mathcal{E}_{E > E^*} - \mathcal{E}_{\{x | V(x) > E^*\}}$ must contain negativities for some energy $E^*$.
\end{proof}

This corollary serves two purposes. First of all, since a Gaussian state is represented by a positive Wigner function over phase space, it is often considered to be a valid joint-probability distribution and as the `least non-classical' state \cite{PSMK}. This corollary serves as a reminder that despite it is true that negativity in Wigner functions as distribution functions is a novel feature in phase space quantum theory, a completely positive Wigner function can still exhibit non-classical behaviours, which, in this specific case, is due to the deformation in effects as discussed in the last section. It has been demonstrated that the sub-theory of Gaussian quantum mechanics can be constructed by imposing certain epistemic restriction on classical phase space mechanics \cite{PSBR}, in which only Gaussian states, measurements and operations are considered. This reinforces the conclusion that the tunnelling rate operator corresponding to the measurements in question is the culprit for a Gaussian state to exhibit tunnelling, for otherwise, the system would simply be classical and cannot tunnel. 


Secondly, the discussion of tunnelling in Gaussian states allows for simplified examples of the application of the previous results, as the tunnelling rate operator $\mathcal{E}_{E > E^*} - \mathcal{E}_{\{x | V(x) > E^*\}}$ determines the tunnelling behaviour of these Gaussian states. Consider the following two examples of Gaussian states:
\\
\\
\textbf{1. Quantum Tunnelling of Wave Packets.}
A particle in quantum mechanics is often described by a wave packet, which has a localised waveform as a superposition of energy eigenstates \cite{QMC}. In particular, Gaussian wave packets are often used in the study of quantum tunnelling and general quantum mechanics problem \cite{QTH}. By studying Gaussian wave packets in rectangular potential barrier, it also introduces dynamical aspects into the problem of tunnelling, as the position probability distribution of a quantum system now changes over time due to the relative phase differences between the different components of its energy eigenstates. 

A Gaussian wave packet that centres at position $x = x_0$ and momentum $p = p_0$, with uncertainty in position as $\Delta x = \sigma_x$, has the form
\begin{equation}
\label{eq:QMGWP}
 \psi(x) = \left(\frac{1}{2\pi\sigma_x^2}\right)^{1/4} \exp{\left[\minus{\frac{(x-x_0)^2}{4\sigma_x^2}}\right]} \exp{\left(\frac{ip_0 x}{\hbar}\right)},
\end{equation}
which has a Gaussian distribution over position space,
\begin{equation}
 P(x) = \sqrt{\frac{1}{2\pi\sigma_x^2}} \exp{\left[\minus{\frac{(x-x_0)^2}{2\sigma_x^2}}\right]}.
\end{equation}

\begin{figure}[!htb]
  \centering
  \includegraphics[scale=0.4]{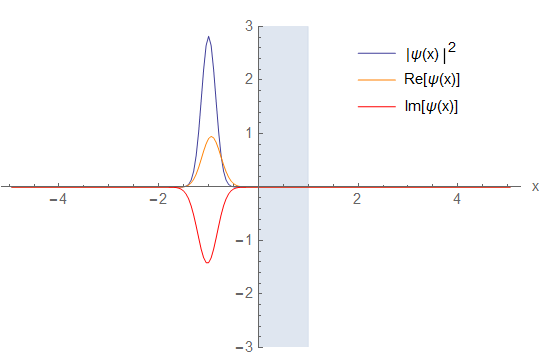}
  \includegraphics[scale=0.4]{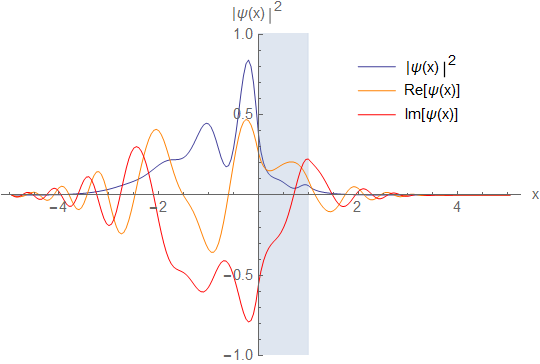}
  \caption{The probability distribution of position of a Gaussian wave packet before (left) and after (right) `passing' the barrier. \label{fig:QMRPBG}}
\end{figure}

It is possible to solve for the dynamics of the Gaussian wave packet as the superposition of time-dependent energy eigenstates of the rectangular potential barrier, or via numerical simulation of Schr\"{o}dinger equation. An example of such a simulation is shown in \text{Fig. \ref{fig:QMRPBG}}, where a wave packet, with initial Gaussian shape and average energy lower than the potential height, passes through the barrier, and the Gaussian nature of the wave packet is destroyed.

\begin{figure}[!htb]
	\centering
	\includegraphics[scale=0.4]{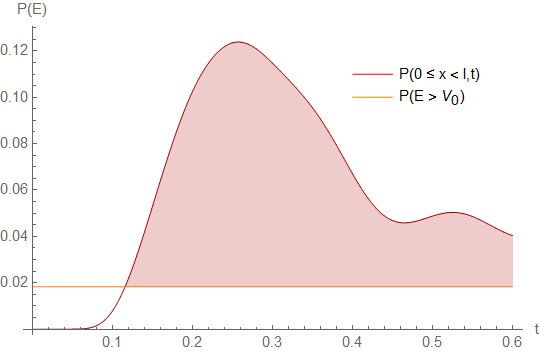}
	\caption{Plot of probabilities $P(0\leq x < l)$ (red) and $P(E > V_0)$ (orange) for a Gaussian wave packet tunnelling through a barrier as a function of time $t$. At times that $P(0\leq x < l, t) \geq P(E>V_0)$, the state is considered to be tunnelling. \label{fig:GWPT}}
\end{figure}

The fact that there is no definite energy for such Gaussian wave packet creates difficulties in applying the standard definition of tunnelling to this case. However, using our general definition, it could be demonstrated that such a state is indeed tunnelling at certain times during the propagation of the wave packet, as $\exists E^*: P(E>E^*) > P(x|V(x) > E^*)$ as shown in \text{Fig. \ref{fig:GWPT}}. The detailed calculations could be found in Appendix D.
\\
\\
\textbf{2. Simple Harmonic Oscillator.}
A simple harmonic oscillator has a potential of the form
\begin{equation}
 V(x) = \frac{1}{2} m \omega^2 x^2,
\end{equation}
where $m$ is the mass of the particle, and $\omega$ is the \textit{angular frequency} of the oscillator. A simple harmonic oscillator is one of the most well-studied potentials in physics, and has many nice features that are exploited in both classical and quantum theories. 

A particle as a classical simple harmonic oscillator will carry out sinusoidal oscillatory motion with \cite{CMK}
\begin{equation}
\label{eq:QMCHO}
 x(t) = C \cos{(\omega t + \phi)},
\end{equation}
where the amplitude $C$ and the phase $\phi$ depends on the initial conditions of the particle. By conservation of energy, a particle with energy $E$ is only allowed to be in the region $-\sqrt{2E/m\omega^2} \leq x \leq \sqrt{2E/m\omega^2}$, for otherwise the particle would have negative kinetic energy.

In the quantum case, in order to solve for the energy eigenstates of the Hamiltonian 
\begin{equation}
\label{eq:QMQHOH}
 \hat{H} = \frac{\hat{P}^2}{2m} + \frac{1}{2} m \omega^2 \hat{X}^2,
\end{equation}
one can use the \textit{ladder operator} method by defining the \textit{creation operator} $\hat{a}^\dagger$ and \textit{annihilation operator} $\hat{a}$ as \cite{QMC}\cite{QMG}\cite{QMS}
\begin{align}
 \hat{a} = \sqrt{\frac{m \omega}{2 \hbar}} \left( \hat{X} + \frac{i}{m \omega} \hat{P} \right),\\
 \hat{a}^\dagger = \sqrt{\frac{m \omega}{2 \hbar}} \left( \hat{X} - \frac{i}{m \omega} \hat{P} \right),
\end{align}
which gives a commutation relation $[\hat{a},\hat{a}^\dagger] = 1$. With the two operators, the Hamiltonian can be rewritten as 
\begin{equation}
 \hat{H} = \hbar \omega (\hat{a}^\dagger \hat{a} + \frac{1}{2}).
\end{equation}
One can identify the operators $\hat{a}^\dagger \hat{a}$ as the \textit{number operator} $\hat{n}$, and label the energy eigenstates as $\ket{n}$ with energy
\begin{equation}
 E_n = \left(n + \frac{1}{2} \right) \hbar \omega,
\end{equation}
and the operations of the annihilation and creation operators on the energy eigenstates are
\begin{align}
 \label{eq:QMQHOAn}
 \hat{a}\ket{n} &= \sqrt{n} \ket{n-1} \\
 \label{eq:QMQHOCn}
 \hat{a}^\dagger \ket{n} &= \sqrt{n+1} \ket{n+1}.
\end{align}

\begin{figure}[!htb]
  \centering
  \includegraphics[scale=0.4]{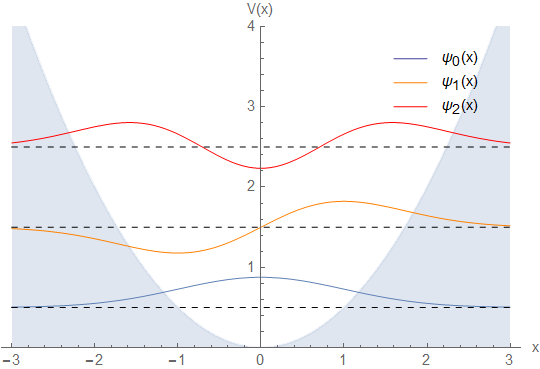}
  \includegraphics[scale=0.4]{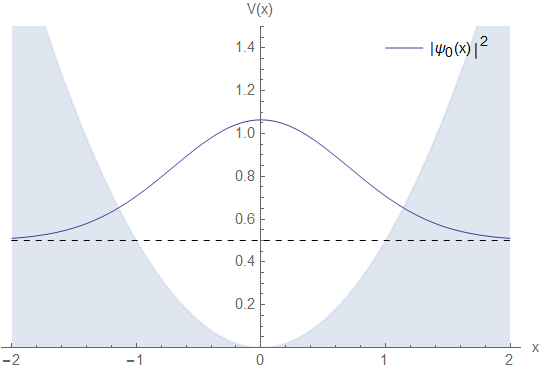}
  \caption{First three eigenstates of quantum harmonic oscillator (left) and probability of locating the ground state of quantum harmonic oscillator (right). The dashed line represent the zeros of the corresponding wave functions and probabilities, and the blue region specifies the classically forbidden region. \label{fig:QMQHO}}
\end{figure}

In this example, the wave function of interest is the ground state of quantum harmonic oscillator, $\psi_0(x) = \braket{x|0}$, which can be solved by the equation $\braket{x | \hat{a}|0} =0$. Solving the equation would give the solution: 
\begin{equation}
\label{eq:QMQHOGS}
 \psi_0(x) = \left(\frac{m\omega}{\pi\hbar}\right)^{1/4} \exp{\left(\minus{\frac{m \omega x^2}{2\hbar}}\right)},
\end{equation}
which has a Gaussian waveform. Since this energy eigenstate has energy $E_0 = \hbar \omega / 2$, and the Gaussian wave function has non-zero amplitudes over the entire position space, there is non-zero probability of finding the ground state in classically forbidden region where $V(x) > \hbar \omega / 2$, hence $P(x|V(x)>E_0) > P(E>E_0)$, satisfying the criteria of tunnelling. This phenomenon also applies to excited states, which can be clearly seen in \text{Fig. \ref{fig:QMQHO}}. 

To mathematically develop the statement about whether the ground state is indeed tunnelling from the viewpoint of the phase space, one could refer to \text{Cor. \ref{cor:PTG}} and consider the operator $\mathcal{E}_{\{E > E^*\}} - \mathcal{E}_{\{x|V(x)>E^*\}}$ at $E^* = E_0 = \hbar \omega / 2$. It can be shown that, in Appendix E, that such an operator for the quantum case is calculated to be 
\begin{align}
	&\,[\mathcal{E}_{E>E_0} - \mathcal{E}_{\{x|V(x) > E_0\}}](x,p) \nonumber \\
	=&\,
	\begin{cases}
		-2 \exp \left(-\frac{m\omega x^2}{\hbar} \right)\exp\left(-\frac{p^2}{m\omega\hbar}\right) & \text{for } V(x) > E_0, \\
		1 - 2\exp \left(-\frac{m\omega x^2}{\hbar}\right) \exp \left(-\frac{p^2}{m\omega\hbar}\right) & \text{otherwise,}
	\end{cases}
\end{align} 
which contains negativities, indicating the ground state of a quantum harmonic oscillator can tunnel. To demonstrate such a state is indeed tunnelling, one could calculate the inner product between this tunnelling rate operator and the Wigner function of the ground state, which is indeed negative as shown in Appendix E. By \text{Thm. \ref{def:QM-GD}}, the ground state of a quantum harmonic oscillator is a tunnelling state. 

\subsection{Wigner function as example of generalised probabilistic theory} In order to demonstrate that phase space quantum theory can be structured under the generalised probabilistic theory framework, it is necessary to recover the mathematical objects corresponding to the three components of a GPT:
\\
\\
\textbf{1. Preparation.} A state of a one-spatial dimensional quantum state can be described by a Wigner function $W(x,p)$, which is essentially the Weyl transformation of a density operator representing that quantum state in Hilbert space formulation. However, there may be difficulties in designating such state as a real vector, which is demanded by the framework of GPT, for mathematical manipulations may not be well-behaved for objects in infinite-dimensional spaces \cite{QMG}.
 
 Nonetheless, real functions $f(x,p)$ in the phase space can be considered as vectors in a vector space, with the vector addition as addition of functions and scalar multiplication as multiplying a scalar \cite{MPH}, in which the set of valid phase space functions, or the state space, is a subset. In addition, one can define an inner product on such vector space such that, for any two functions $\Omega_1$ and $\Omega_2$:
 \begin{equation}
  \braket{\Omega_1,\Omega_2} = \iint \Omega_1(x,p) \Omega^*_2(x,p)\,dx\,dp,
 \end{equation}
 such that a Wigner function corresponding to a pure state is square-integrable,
 \begin{equation}
  \iint |W(x,p)|^2 \,dx \,dp = \frac{1}{2\pi\hbar},
 \end{equation} 
 which implies that the function converges and the inner product exists for these states as a sign of well-behaved theory. Moreover, for general normalised states that is represented by $\hat{\rho} = \sum_i \lambda_i \ket{\psi_i}\bra{\psi_i}$ with $\sum_i \lambda_i = 1$ and $0 \leq \lambda_i \leq 1$, the Wigner representation is
 \begin{align}
  &\,\frac{1}{\pi \hbar} \int e^{2ipy/\hbar} \braket{x-y|\sum_i \lambda_i |\psi_i}\braket{\psi_i|x+y}\,dy \nonumber \\
  =&\, \sum_i \frac{\lambda_i}{\pi\hbar} \int e^{2ipy/\hbar} \psi_i^*(x+y) \psi_i(x-y)\,dy \nonumber \\
  =&\, \sum_i \lambda_i W_i(x,p),
 \end{align}
 where $W_i(x,p)$ are Wigner functions corresponding to the pure states $\ket{\psi_i}\bra{\psi_i}$. In this case, since $\sum_i \lambda_i = 1$, these Wigner functions corresponds to the normalised states, and therefore any normalised is a convex sum of the pure states $W_i$. Therefore, it is demonstrated that the Wigner functions corresponding to the pure states are the extremal states of the GPT, and the set of normalised states is the convex hull of such extremal states. Therefore, these normalised mixed states are square-integrable as well, as 
 \begin{align}
  &\,\iint W^2(x,p)\,dx\,dp \nonumber \\
  =&\,\iint \left[\sum_i \lambda_i W_i(x,p)\right]^2\,dx\,dp \nonumber \\
 =&\,\sum_{i,j} \lambda_i \lambda_j \iint W_i(x,p) W_j(x,p) \,dx \,dp \nonumber \\
 \leq&\, \sum_i  \frac{\lambda_i^2}{2\pi\hbar} + \sum_{i \neq j} \frac{\lambda_i\lambda_j}{2\pi\hbar} \nonumber \\
 =&\,\frac{1}{2 \pi \hbar} \left(\sum_i \lambda_i\right)^2 =\frac{1}{2 \pi \hbar},
 \end{align}
 which implies $1/2\pi\hbar$ is the maximum of the square-integral of any Wigner function. In fact, such result is related to the purity measure $\Tr(\hat{\rho}^2) = \mu$. This measure can be mapped to
 \begin{equation}
 \label{eq:GPSP}
  \Tr(\hat{\rho}^2) = 2 \pi \hbar \iint W^2(x,p)\,dx\,dp = \mu,
 \end{equation}
 where $0 \leq \mu \leq 1$, with $0 \leq 2\pi\hbar \iint W_i(x,p) W_j(x,p) \leq 1$. This implies the square integral, or the inner product of the state on itself, or geometrically the square of the length of the state vector, is linearly related to purity. An interesting comparison to be made here is that the length of a vector in Bloch sphere corresponds also to the purity measure of a qubit, and that in quantum theory of qubit, such length is also related to uncertainty of a state.
 
 Sub-normalised states can be generated by rescaling a normalised Wigner function with some non-negative constant $\nu \leq 1$, and therefore in general, any state can be expressed as  
 \begin{align}
  W(x,p) = \nu \sum_i \lambda_i W_i(x,p) = \sum_i \lambda_i [\nu W_i(x,p)],
 \end{align}
 as the convex sum of extremal states and the null vector. Therefore, the states of phase space formulation concur with the requirements of the GPT framework.
 
 There is, however, one caveat: the conventional Wigner representation of a position eigenstate is linearly related to a Dirac delta function $\delta(x-x_0)$. This state is not well-behaved in the sense that the square-integral of the function diverges, causing problems in defining the inner product of such vector space as the integral of two phase space functions. Nonetheless, it should be noted that the same problem arises in the Hilbert space formulation of quantum theory, where the wave function of a position eigenstate is not considered to be a state in Hilbert space and therefore does not correspond to a physical state \cite{QMG}. Therefore, the introduction of such states in the GPT framework does not lead to additional non-physical features other than the ones inherent in the conventional quantum theory \cite{QMG}.
 
 As a side remark, a possible method to bypass this problem by discretising the Wigner function into an $N \times N$ grid in phase space, in which each grid at $(x_i, p_j)$ represents the quasi-probability of locating the state between some range $x_i \leq x < x_i + \Delta x$ and $p_j \leq p < p_j+\Delta p$, i.e. $\int_{p_j}^{p_j + \Delta p}\int_{x_i}^{x_i+\Delta x} W(x,p) \,dx\,dp$. Even though it would require $N \rightarrow \infty$, it would regularise the Dirac delta functions $\delta(x-x_0)$ into having the value $1/N$ at each grid where $x_i \leq x_0< x_i + \Delta x$, instead of infinities. Since these quasi-probabilities are real numbers, at $N \rightarrow \infty$ a scaled Wigner function is recovered, the GPT formulation can be applied to such object as representation of a quantum state. Needless to say, this process would require a detailed formulation of discrete Wigner functions and phase space quantum mechanics.
 \\
 \\
 \textbf{2. Transformation.} The set of valid transformations is the set of functions with star-product as the Weyl transform of transformations in Hilbert space formulation. An example of such transformation is the unitary evolution of state. The Weyl transformation of such evolution is given by
 \begin{equation*}
  \hat{U} \rho \hat{U}^\dagger \rightarrow U(x,p) \star W(x,p) \star U^*(x,p),
 \end{equation*}
 where $U(x,p)$ is the Weyl transform of the unitary operator $\hat{U}$. It can be demonstrated that complex conjugate transposition of an operator in Hilbert space formulation is mapped to a complex conjugation of the phase space function, as 
 \begin{align*}
  &\,(U^\dagger)(x,p) \\
  =&\, 2\int_{\minus \infty}^\infty e^{2ipy/\hbar} \braket{x+y|U^\dagger|x-y}\,dy \\
  =&\, \minus 2\int_{\infty}^{\minus \infty} e^{\minus 2ipy/\hbar} \braket{x-y|U^\dagger|x+y}\,dy \\
  =&\, \left(2 \int_{\minus \infty}^\infty e^{2ipy/\hbar} \braket{x+y|U|x-y}\,dy\right)^* \\
  =&\, U^*(x,p).
 \end{align*}
 Therefore, the unitary condition for a reversible transformation is expressed as 
 \begin{equation*}
  \hat{U}^\dagger \hat{U} = \mathbb{1} \rightarrow U^*(x,p) \star U(x,p) = 1.
 \end{equation*}
 It is also interesting to note that unitary transformations preserve purity by
 \begin{align*}
  \Tr\left([\hat{U} \hat{\rho} \hat{U}^\dagger]^2\right) = \Tr(\hat{U}\hat{\rho}^2\hat{U}^\dagger) = \Tr(\hat{U}^\dagger \hat{U} \hat{\rho^2} = \Tr(\hat{\rho}^2).
 \end{align*}
 In the phase space formulation, this is equivalent to the case where the unitary transformations preserve the length of the state vector. In general,  unitary transformations preserve the inner products between two states, and therefore in the phase space picture such transformations resembles an orthogonal transformation, despite the dimension of the state space is infinite. This feature is analogous to the isomorphism between $SU(2)$ and $SO(3)$ in the Bloch sphere representation of quantum theory of qubits.
 
 In order to demonstrate that transformations in phase space formulation of quantum theory concurs with that of a GPT, it is necessary to demonstrate that such transformations are linear. Since it is the case in Hilbert space formulation that a transformation, represented by a CPTP map as $\hat{\Phi}$, must be linear, by Weyl transformation,
 \begin{align}
  \hat{\Phi}(\sum_i \lambda_i \hat{\rho_i}) &= \sum_i \lambda_i \hat{\Phi}(\hat{\rho}_i) \nonumber \\ 
  \Phi(x,p) \star \sum_i \lambda_i W_i(x,p) &= \sum_i \lambda_i \Phi(x,p) \star W_i(x,p), \nonumber
 \end{align}
 where $\Phi(x,p)$ is the operator in phase space representation, and $W_i(x,p)$ are Wigner functions corresponding to density operators $\hat{\rho}_i$. This can also be seen from the construction of star-product \eqref{eq:star}, where the fact that the operators $\vec{\partial}_x$ and $\vec{\partial}_p$ are linear implies that the star-product as a expansion of these operators is linear as well, and therefore is consistent with the formulation of transformation in a GPT.
 
 It is also relatively straightforward to demonstrate that such transformations map any valid state into another by directly converting such result from the Hilbert space formulation via Weyl transformation. Since it is the case that each state in Hilbert space formulation can be mapped to a Wigner function as a valid state, and that each CPTP map is a transformation that maps one valid state into another, it must be the case that these transformations in phase space formulation map any state into valid states. 
 \\
 \\
 \textbf{3. Measurement.} The set of valid effects in the phase space formulation is the set of functions as the Weyl transformation of \textit{positive operator valued measurement}, or \textit{POVMs}, in Hilbert space formulation. In particular, the set of pure effects $\mathcal{E}_i$ are the phase space functions corresponding to projective measurement in the form $\hat{\Pi}_i = \ket{\omega_i}\bra{\omega_i}$, such that
 \begin{align*}
  &\, \mathcal{E}_i(x,p)\\
 = &\, 2\int e^{2ipy/\hbar} \braket{x-y | \hat{\Pi}_i | x+y} \,dy \\
  = &\, 2\int e^{2ipy/\hbar} \braket{x-y| \omega_i}\braket{\omega_i|x+y}\,dy,
 \end{align*}
 which is isomorphic to the set of pure states. With the POVMs as probabilistic mixture of projection operators, any mixed effects can be expanded as a convex sum of pure effects and null effect. From this result, the self-duality of state space and effect space, a property of quantum theory, is recovered. 
 
 Notice that under the Weyl transformation, a probability of some event corresponding to $\omega_i$ occurring is given by the inner product of the effect $\mathcal{E}_i$ and the state $W$,
 \begin{equation*}
  \Tr(\hat{\Pi}_i \hat{\rho}) \rightarrow \iint \mathcal{E}_i(x,p) W(x,p)\,dx\,dp.
 \end{equation*}
Here, the probabilities satisfy $0 \leq \Tr(\hat{\Pi}_i \hat{\rho}) \leq 1$. These projectors corresponding to a measurement summing to identity by the completeness equation, and the set of orthogonal eigenstates is mapped to a set of orthogonal functions in phase space as shown in Appendix F. Therefore, for some projection operators $\hat{\Pi}$ corresponding to multiple outcomes, $ 0 \leq \Tr(\hat{\Pi} \hat{\rho}) \leq \Tr(\mathbb{1} \hat{\rho}) \leq 1$. Hence, for general states and general effects, the inner product has the upper bound
 \begin{align}
  & \Tr\big(\sum_i \mu_i \hat{\Pi}_i \sum_j \lambda_j \ket{\psi_j}\bra{\psi_j}\big) \nonumber \\
 = &\, \sum_{i,j} \mu_i \lambda_j \Tr(\hat{\Pi}_i\ket{\psi_j}\bra{\psi_j}) \nonumber \\
  \leq&\, \sum_{i,j} \mu_i \lambda_j \leq \big(\sum_i \mu_i\big)^2 \big(\sum_j \lambda_j\big)^2 \leq 1,
 \end{align}
 by Cauchy-Schwarz inequality and the condition $\sum_i \mu_i \leq 1$ and $\sum_j \lambda_j \leq 1$. Such inner product is also non-negative, as $0 \leq \mu_i, \lambda_j \leq 1$ and that the trace of pure states and pure effects must be non-negative. By mapping these results into the phase space formulation, these effects satisfy the requirement of the GPT framework as to give valid probability values under inner product with any states in state space.

\subsection{Tunnelling in Post-Quantum Theories} To study the phenomenon of tunnelling in post-quantum theories, one will have to devise a method of consistently extending the existing state space in phase space picture into inclusion of non-physical states. Extension of state space is relatively easy with Wigner function representation, as by varying the values of the Wigner function at some phase space points, it is possible to generate a non-physical state. The difficulty, however, lies in construction of the corresponding effect space, and the physical interpretation of these post-quantum states. While a complete post-quantum theory is not devised in this Article, some preliminary work on Gaussian states is carried out as a precursor towards an eventual development of a post-quantum phase space theory.

A Gaussian bivariate distribution in phase space $W_G(x,p)$ has the form \cite{PSBR}
\begin{equation}
\label{eq:GPSG}
 W_G(x,p) = \frac{1}{2\pi \det \gamma^{1/2}} \exp \left[\minus \frac{1}{2} (\vec{x} - \vec{\mu})^T \gamma^{-1} (\vec{x}-\vec{\mu})\right],
\end{equation}
where $\vec{x}$ is a vector of coordinates, $\vec{\mu}$ is a vector of mean values of coordinates, and $\gamma$ is the covariance matrix, where $\gamma_{ij}$ is the covariance of $i$-th and $j$-th coordinates. A possible way of generalising the existing state space is to include Gaussian distributions that violates the purity condition \eqref{eq:GPSP}, such that
\begin{equation}
\label{eq:GPSGP}
 2 \pi \hbar \iint W_G(x,p) \,dx \,dp > 1.
\end{equation}
The main advantage of considering only Gaussian distributions as post-quantum states is that it is a proper joint-probability distribution in phase space, and therefore such an object gives proper probabilities when one conducts measurement in position or momentum space. Another advantage is that such a distribution is positive over all phase space, and therefore, by \text{Cor. \ref{cor:PTG}}, any analysis of tunnelling on these post-quantum states are dependent only on the effects $\mathcal{E}_{E > E^*}$ and $\mathcal{E}_{\{x|V(x) > E^*\}}$, provided that the effects remain valid under such an extension of state space.

It is interesting to note the physical meaning behind extension of purity. One way of understanding such an extension is to calculate the variances in both position and momentum for $W_G(x,p)$. Since it is the property of a Gaussian bivariate distribution to yield a Gaussian distribution in one coordinate after integrating the distribution over the other coordinate, the variances in $x$ and $p$ are simply $\sigma_x^2 = \gamma_{xx}$ and $\sigma_p^2 = \gamma_{pp}$. For a state with given purity $\mu$,
\begin{align}
\label{eq:GPSPG}
 \mu &= 2 \pi \hbar \iint W_G^2(x,p)\,dx\,dp \nonumber \\
 &= \frac{2 \pi \hbar}{(2\pi)^2 (\det \gamma^{1/2})^2} \iint \exp\left[\minus (\vec{x} - \vec{\mu})^T \gamma^{-1} (\vec{x} - \vec{\mu})\right]\,dx\,dp \nonumber \\
 &= \frac{2\pi\hbar}{4 \pi \sqrt{\sigma_x^2\sigma_p^2-\gamma_{xp}^2}} \geq \frac{\hbar}{2\sqrt{\sigma_x^2\sigma_p^2}} \nonumber \\
 \sigma_x \sigma_p &\geq \frac{\hbar}{2\mu}.
\end{align}
since $\gamma_{xp}^2 \geq 0$. The final result of the calculation closely resembles the uncertainty principle. In fact, if one substitutes $\mu = 1$ as the condition of pure state, the uncertainty principle is recovered. Therefore, by relaxing the purity relation to states that has $\mu > 1$, the lower bound of the uncertainties $\sigma_x \sigma_p$ reaches below the lower bound allowed by quantum theory, and therefore violates the quantum uncertainty principle. This matches the intuition on an extreme case where a Dirac delta function $\delta(x-x_0)\delta(p-p_0)$ represents a distribution with perfect information of both position and momentum, is a state with infinite purity $\mu$ since $\iint \delta^2(x-x_0)\delta^2(p-p_0)\,dx\,dp \rightarrow \infty$, which maximally violates the uncertainty principle. 

It should be noted that certain studies had related purity and uncertainty before \cite{GPSV}; yet the results are limited to purity being between $0 \leq \mu \leq 1$. Nonetheless, the practical implication of this result is that one can simply change the purity condition and generate post-quantum states according to \eqref{eq:GPSGP} by altering the covariant matrix $\gamma$ in \eqref{eq:GPSG}. 

Another way of interpreting the violation of purity condition is to look at the Hilbert space formulation of quantum mechanics. For a general state $W(x,p)$ such that it is real, it is mapped by Wigner transformation to an operator $\hat{\rho}$ such that
\begin{align}
 &\,\hat{\rho}^\dagger \nonumber \\
 =&\, \frac{1}{(2\pi)^2} \left[\iiiint W(x,p) e^{i[\alpha(\hat{X} - x) + \beta(\hat{P} - p)]}\,d\alpha\,d\beta\,dx\,dp \right]^\dagger \nonumber \\
 =&\, \frac{1}{(2\pi)^2} \iiiint W^*(x,p) e^{\minus i[\alpha(\hat{X} - x) + \beta (\hat{P} - p)]}\,d\alpha\,d\beta\,dx\,dp \nonumber \\
 =&\, \frac{1}{(2\pi)^2} \iiiint W(x,p) e^{i[(\minus \alpha)(\hat{X} - x) + (\minus \beta)(\hat{P} - p)]} \nonumber \\
 &\,d(\minus \alpha)\,d(\minus \beta)\,dx\,dp \nonumber \\
=&\, \hat{\rho}.
\end{align}
In other words, the reality condition of a Wigner function is mapped to the Hermiticity of the density operator. Since it is possible to find an eigen-decomposition for any Hermitian operators \cite{QMS}, one can write a density operator as 
\begin{equation}
 \hat{\rho} = \sum_i \lambda_i \ket{\psi_i} \bra{\psi_i},
\end{equation}
where $\lambda_i$ are eigenvalues corresponding to $\ket{\psi_i}$ as $i$-th eigenvector. In this representation, by the orthogonality of eigenvectors, purity is simply
\begin{equation}
 \Tr(\hat{\rho}^2) = \Tr\big(\sum_i \lambda_i^2 \ket{\psi_i}\bra{\psi_i}\big) = \sum_i \lambda_i^2.
\end{equation}
Therefore, for a general normalised post-quantum state with purity $\mu > 1$ represented by a real Wigner function, it can be mapped to a density operator in its eigenbasis $\hat{\rho} = \sum_i \lambda_i \ket{\psi_i} \bra{\psi_i}$, such that
\begin{align}
 \sum_i \lambda_i = 1; \label{eq:GPSN}\\
 \sum_i \lambda_i^2 > 1. \label{eq:GPSPP}
\end{align}
Assume that all $\lambda_i \geq 0$. By some algebraic manipulation,
\begin{align}
 \big(\sum_i \lambda_i \big)^2 = \sum_i \lambda_i^2 + \sum_{i \neq j} \lambda_i \lambda_j \leq \sum_i \lambda_i^2 = 1.
\end{align}
Therefore, it is impossible to satisfy the both \eqref{eq:GPSN} and \eqref{eq:GPSPP} together with the assumptions $\lambda_i \geq 0$. By reductio ad absurdum, it is necessary that the density operator corresponding to post-quantum states with purity greater than unity to have negative eigenvalues. In other words, in the Hilbert space formulation, these post-quantum states violates quantum theory by introducing non-positive definite operators as states. This serves as a warning of altering the state space without correspondingly changing effect space, as this would introduce observable non-physical probabilities beyond the conventional range between 0 and 1.

With the procedure \eqref{eq:GPSGP} ultimately related to introducing negativities into the density operator, there is an obvious problem with the extension towards post-quantum theory. It can no longer be conceived that the set of measurements is invariant under such alteration, for to do so is to allow negative probabilities when one conduct an inner product of an effect corresponding to the eigenstate with negative eigenvalue and the state. Therefore, the set of allowed measurements must shrink accordingly. Despite it is true that by the construction of \eqref{eq:GPSG}, the effect $\mathcal{E}_{\{x | V(x)>E^*\}}$ is still valid, it is not so apparent that $\mathcal{E}_{E > E^*}$ remains valid. If in certain states that the set of effects corresponding to energy measurement is invalid, then one must find another set of effects corresponding to a new energy measurement, which creates great difficulties in interpreting energy as a physical and observable quantity. 

Nonetheless, a qualitative argument could be given here regarding the status of tunnelling as a phenomenon in post-quantum scenarios. Consider the ground state of the quantum harmonic oscillator \eqref{eq:QMQHOGS}: if one varies the state by shrinking $\sigma_x$ and $\sigma_p$ simultaneously, and therefore violating the uncertainty principle and purity condition by altering the covariance matrix $\gamma$ in \eqref{eq:GPSG}, then while the ground state energy effect vector would shrink correspondingly, the variation is continuous and therefore the inner product that specifies the tunnelling rate would still retains negativity in the close vicinity of the quantum case. Therefore, it seems that tunnelling can be a generic property of post-quantum theories and is not unique to quantum theory. However, to fully study the phenomenon of tunnelling in post-quantum theories rigorously, it is necessary to construct a systematic theory that describes the effects on the effect space by alteration of the state space.

\section{Summary}

 We have shown that tunnelling necessitates a negative Wigner function of the state and/or a tunnelling rate operator at some energies as we have defined. This links tunnelling with non-classical probabilistic behaviour (negative quasi-probabilities) in a concrete manner. We also argued that our approach can be used to investigate tunnelling in generalised probabilistic theories, showing the Wigner function representation fits into that framework. 

A very intriguing question for future research is how these results relate to recent studies that suggest negative Wigner function CITE (which is also argue to be equivalent to contextuality CITE) is the 'source' of the putative power of quantum computation \cite{QMWC}. 

\section{Acknowledgements}
We are grateful for discussions with Dan Browne, Jonathan Halliwell, Benjamin Yadin, Andrew Garner, Vlatko Vedral, and Dominic Branford. We acknowledge funding from the EU collaborative project TherMiQ (Grant agreement No.~618074), Wolfson College, University of Oxford and the London Institute for Mathematical Sciences.


\onecolumngrid
\appendix
\setcounter{secnumdepth}{1}

\section{Reflection over a barrier}
The general definition \text{Def. \ref{def:QM-GD}} could be thought of as a case of non-classical behaviour in position space. A corresponding non-classical behaviour known as reflection over barrier, could be considered under the same framework as tunnelling in momentum space. Following the same analysis, the classically forbidden region for a particle with energy $E^*$ is $\{p | |p| < \sqrt{2m(E^* - \sup_\mathbb{R} V)} \}$, denoted as $\mathcal{P}(E^*)$, where $p$ is a real variable representing momentum. Such constraint also applies to states with energy $E > E^*$, or that $\mathcal{P}(E^*) \subset \mathcal{P}(E)$. Therefore, one can formulate the definition of reflection over barrier as:

\begin{thmdef}
\label{def:QM-RB}
 \textbf{General Definition of Reflection over Barrier.} For a state in a potential given by $V(x)$, it is reflecting over barrier if and only if there exists some energy $E^*$, such that the probability of locating the state in region where $|p| < \sqrt{2m(E^*-\sup_\mathbb{R} V)}$ is greater than that of measuring the state to have energy $E < E^*$, or mathematically,
 \begin{equation*}
  \exists E^*: P\left(p \middle| |p| < \sqrt{2m(E^*-\sup_\mathbb{R} V})\right) > P(E < E^*).
 \end{equation*}
\end{thmdef}

It should be noted that the two definitions, \text{Def. \ref{def:QM-GD}} and \text{Def. \ref{def:QM-RB}}, has similar structure, and this provides an example of constructing a mathematical formulation of various non-classical quantum processes: starting from certain classical relations between physical quantities, such as position $x$ and energy $E$ in \text{Def. \ref{def:QM-GD}}, one could define some classically forbidden region $\mathcal{X}(E^*)$ which relates the two quantities, and any states that violates this relation are considered to be a non-classical state. 

\section{Recovery of Standard Definition}

The general definition can be applied to an energy eigenstate solution with energy $E_0$ of a rectangular potential barrier with the form \eqref{eq:RPB} with potential height $V_0$ and length $l$ and provide a condition definition of tunnelling in such case. Mathematically, the two probabilities in the general definition becomes
\begin{align}
P(x | V(x) > E^*) &= \begin{cases}
1 & \text{for }E^* = 0 \\
P(x | V(x) = V_0) = P(0 \leq x < l) & \text{for }0 < E^* < V_0 \\
0 & \text{for } E^* \geq V_0;
\end{cases} \label{eq:DTRPBEXE}
\\
&\nonumber\\
P(E > E^*) &= \begin{cases}
1 & \text{for }0 \leq E^* < E_0 \\
0 & \text{for }E^* \geq E_0.  
\end{cases} \label{eq:DTRPBEEE}
\end{align}
Since, by \text{Def. \ref{def:QM-GD}}, a state is tunnelling if and only if there exists some energy $E^*$ such that $P(x | V(x) > E^*) > P(E>E^*)$, the general definition is equivalent to:

\begin{thmdef}
	\label{def:QM-ED}
	\textbf{Equivalent General Definition of Tunnelling for Energy Eigenstates in Rectangular Potential Barrier.} An energy eigenstate with energy $E_0$ in a rectangular potential barrier of the form \eqref{eq:RPB} is tunnelling if and only if $E_0 < V_0$ and the probability of finding the state in the region $0 \leq x < l$ is non-zero.
\end{thmdef}

While this equivalent statement has much similarities with the conventional definition \text{Def. \ref{thm:SD}}, a crucial distinction is that the convention definition provides a less strict condition of having non-zero probability of locating the state in $x \geq 0$. This does illustrate an important distinction between the rationale behind the formulation of the general and conventional definition. The reason why the conventional definition includes the region $x \geq l$ to be classically forbidden region for a state with $E_0 < V_0$ is that it assumes the energy eigenstate to be incoming from $x \rightarrow \minus{\infty}$, and therefore cannot classically pass through the barrier. On the other hand, the classically forbidden region of general definition does not take the initial condition of the state into account, so in this sense it can be equally applied to cases where the state is incoming from $x \rightarrow \minus{\infty}$ and $x \rightarrow \infty$. In this sense, the general definition can be applied to a broader classes of states.

Undeniably, however, the equivalent definition \text{Def. \ref{def:QM-ED}} provides a stricter definition of tunnelling, and hence is not immediately an equivalent statement with the conventional definition \text{Def. \ref{thm:SD}}. However, one can justify that, specifically in the quantum case, the two definitions are equivalent with the inclusion of the continuity condition of wave functions. With this condition, a quantum energy eigenstate can only have non-zero probability in the region $x \geq l$ if there is non-zero probability distribution of locating the state at $0 \leq x < l$, which implies the condition $P(0 \leq x < l) > 0 \rightarrow P(x \geq l) > 0$, assuming the state is incoming from $x \rightarrow \minus{\infty}$. Therefore, the general definition and the continuity condition of wave function leads to the recovery of the conventional definition.

\section{Tunnelling in Classical Phase Space Distributions}
\label{sec:TPS-C}
With the framework of tunnelling in phase space stated in \text{Thm. \ref{thm:PSTNS}} and \text{Cor. \ref{col:PSTN}}, it is straightforward, then, to demonstrate that it is impossible for classical systems to tunnel in the following manner:

\begin{thm}
	\label{thm:PSTC}
	\textbf{Impossibility of Tunnelling in Classical Systems.} A classical system, specified with the Hamiltonian $H(x,p)$ and the distribution function $f(x,p)$, cannot tunnel. 
\end{thm}

\begin{proof}
	For a general classical system, the effect corresponding to probability $P(x | V(x) > E^*)$ is 
	\begin{equation*}
	\mathcal{E}_{\{x|V(x) > E^*\}} (x,p) = \begin{cases}
	1 & \text{for } V(x) > E^* \\
	0 & \text{otherwise,}
	\end{cases}
	\end{equation*} 
	and similarly, the effect corresponding to probability $P(E > E^*)$ is 
	\begin{equation*}
	\mathcal{E}_{E > E^*} (x, p) = \begin{cases}
	1 & \text{for } H(x,p) > E^* \\
	0 & \text{otherwise.}
	\end{cases}
	\end{equation*}
	However, since the set $\{(x,p) | V(x) > E^*\} \subset \{(x,p) | H(x,p) > E^*\}$, as 
	\begin{align*}
	(x^*, p^*) \in \{(x,p) | V(x) > E^*\} &\leftrightarrow V(x^*) > E^* \\
	&\rightarrow \frac{{p^*}^2}{2m} + V(x^*) > E^* \\
	&\leftrightarrow H(x^*, p^*) > E^* \\
	& \leftrightarrow (x^*, p^*) \in \{(x,p) | H(x,p) > E^*\},
	\end{align*}
	therefore $V(x) > E^*$ implies $H(x,p) > E^*$ for a classical system, which suggests that the function $[\mathcal{E}_{\{x|V(x) > E^*\}}-\mathcal{E}_{E > E^*}] (x,p) \geq 0$. Also, since $f(x,p)$ in this scenario is a joint probability function, $f(x,p)$ must be non-negative over all phase space. Therefore, by \text{Cor. \ref{col:PSTN}}, a classical system cannot tunnel. 
\end{proof}

In some sense, this proof is anticipated by the design of the general definition, as part of the original intentions of constructing such a definition. However, this exercise is still valuable, because the previous discussion in the main text is largely based on classical particles rather than phase space ensembles. Another important point is that this proof illustrates is the dual nature of tunnelling in phase space, as tunnelling does not only depend on condition on the state, i.e. the distribution function, but also the difference in effects $\mathcal{E}_{\{x|V(x) > E^*\}}-\mathcal{E}_{E > E^*}$. In particular, both functions have to be non-negative for a state to not tunnel. This concurs with the discussion on \text{Cor. \ref{cor:PTG}}, where a Gaussian state has a positive Wigner function representation, yet in certain scenarios such states could indeed tunnel.

\section{Quantum Tunnelling of Wave Packets}

While it is rather difficult to obtain the energy distribution of a wave packet due to the piecewise nature of the energy eigenstate wave functions, there are certain features of the problem that simplifies the analysis in principle. Firstly, despite the time evolution of the state, the cumulative probability distribution $P(E > E^*)$ is invariant. This is due to the fact that a general state can be considered as a superposition of energy eigenstates $\ket{\psi} = \sum_i c_i \ket{E_i}$, and under the unitary operator $\exp[\minus i \hat{H} t/\hbar]$, 
\begin{align}
e^{\minus i\hat{H}t/\hbar} \ket{\psi} &= \sum_i c_i e^{\minus i \hat{H}t/\hbar} \ket{E_i} \nonumber \\
&= \sum_i c_i e^{\minus iE_it/\hbar} \ket{E_i}.
\end{align}
Therefore, the probability $P(E > E^*)$ is the sum of the norm squared amplitude corresponding to energy eigenstates with energy greater than $E^*$, or
\begin{equation}
P(E > E^*) = \sum_{E_i > E^*} |c_i|^2,
\end{equation}
which is independent of time. Hence, the dynamical nature of the wave packet only changes the probability $P(x | V(x) > E^*)$. 

Secondly, given that for a rectangular potential barrier, $P(x | V(x) > E^*) = P(0 \leq x < l)$ is a constant value between $0 < E^* < V_0$ as shown in \eqref{eq:DTRPBEXE}, it can then be demonstrated that the state is tunnelling if and only if $P(E > V_0) < P(0 \leq x < l)$, by demonstrating that $P(E > V_0) < P(0 \leq x < l)$ is equivalent to the condition $\exists E^*: P(x | V(x) > E^*) > P(E > E^*)$, i.e. the condition of a tunnelling state:

\begin{enumerate}
	\item $\mathbf{\exists E^*: P(x | V(x) > E^*) > P(E > E^*) \rightarrow P(E > V_0) < P(0 \leq x < l)}.$ Assume it is the case that $\exists E^*: P(x | V(x) > E^*) > P(E > E^*)$ and $P(E > V_0) \nless P(0\leq x < l)$. Construct $P (x | V(x) > E^*)$ according to \eqref{eq:DTRPBEXE},
	\begin{align*}
	P(x | V(x) > E^*) &= \begin{cases}
	1 & \text{for } E^* = 0 \\
	P(0 \leq x < l) & \text{for } 0 < E^* < V_0 \\
	0 & \text{for } E^* \geq V_0
	\end{cases}\\
	&\leq \begin{cases}
	1 & \text{for } E^* = 0 \\
	P(E > V_0) & \text{for } 0 < E^* < V_0 \\
	0 & \text{for } E^* \geq V_0
	\end{cases} \\
	&\leq \begin{cases}
	1 & \text{for } E^* = 0 \\
	P(E > E^*) & \text{for } 0 < E^* < V_0 \\
	0 & \text{for } E^* \geq V_0 
	\end{cases} \\
	& \leq P(E > E^*),
	\end{align*}
	where in the region $0 < E^* < V_0$, $P(E > V_0) = P(E > E^*) - P(V_0 > E > E^*) \leq P(E > E^*)$, and in the region $E^* \geq V_0$, $0 \leq P(E > E^*)$. Therefore, for all possible $E^*$, $P(x | V(x) > E^*) \leq P(E > E^*)$, which contradicts with the premise $\exists E^*: P(x|V(X) > E^*) > P(E>E^*)$. Therefore, by reductio ad absurdum, $\exists E^*: P(x | V(x) > E^*) > P(E > E^*) \rightarrow P(E > V_0) < P(0 \leq x < l)$.
	
	\item $\mathbf{P(E > V_0) < P(0 \leq x < l) \rightarrow \exists E^*: P(x | V(x) > E^*) > P(E > E^*).}$ Since $P(0 \leq x < 1)$ is simply $P(x | V(x) > V_0)$,  $P(E > V_0) > P(x | V(x) > V_0)$ implies $\exists E^*: P(x | V(x) > E^*) > P(E > E^*)$. 
\end{enumerate}

What the previous exercise shows is that it suffices to use two probabilities, $P(E > V_0)$ and $P(0 \leq x < l)$ of a state in rectangular potential barrier to determine whether the state is tunnelling or not, which could be applied to our example of Gaussian state tunnelling through a rectangular barrier. Using the approximation that the energy eigenstates are roughly free momentum eigenstates, one can conduct a Fourier transform on \eqref{eq:QMGWP} and obtain the Gaussian wave packet has a probability distribution over momentum space as
\begin{equation}
P(p) = \sqrt{\frac{2}{\pi}} \frac{\sigma_x}{\hbar} \exp{\left[\minus\frac{2\sigma_x(p-p_0)^2}{\hbar}\right]},
\end{equation}
which is also Gaussian, as expected from the property of Fourier transform of Gaussian probability distributions. Therefore, the cumulative energy probability is given by
\begin{align*}
P(E > E^*) &= \int_{\minus \infty}^{\minus \sqrt{2mE^*}}P(p)\,dp + \int_{\sqrt{2mE^*}}^\infty P(p)\,dp \\
&= 1 - \frac{1}{2}\left[\erf\left(\frac{\sqrt{2}\sigma_x}{\hbar}(\sqrt{2mE^*}-p_0)\right)+\erf\left(\frac{\sqrt{2}\sigma_x}{\hbar}(\sqrt{2mE^*}+p_0)\right)\right],
\end{align*}
where $\erf$ is \textit{error function}. While it is difficult to find the close form of the probability $P(0 \leq x < l)$, via some numerical simulation, it is possible to obtain this probability as a function of time, as shown as \text{Fig. \ref{fig:GWPT}}.

Regardless, this analysis demonstrates how the general definition of tunnelling can be used to determine whether a general state as a superposition of energy eigenstates is tunnelling or not, as well as to assign a quantitative value to how such a state violates the classicality constraints. It therefore provides evidence on how the general definition satisfies the criteria of providing quantitative criteria of tunnelling for systems with states without definite energy. 

\section{Quantum Tunnelling of Ground State of Quantum Harmonic Oscillator}

For a quantum harmonic oscillator with the form \eqref{eq:QMQHOH}, the ground state wave function $\psi_0(x)$ is given by \eqref{eq:QMQHOGS}. Such wave function can be mapped to a Wigner function $W_0(x,p)$ by Weyl transformation to be
\begin{equation}
W_0(x,p) = \frac{1}{\pi\hbar} \exp\left(\minus\frac{m\omega x^2}{\hbar}\right)\exp\left(\minus\frac{p^2}{m\omega\hbar}\right).
\end{equation}
Notice here that it can be clearly seen how such Gaussian state saturates the uncertainty principle, as $\sigma_x \sigma_p = \sqrt{\hbar/2m\omega} \sqrt{m\omega\hbar/2} = \hbar / 2$. Also, such a Wigner function, as a map of a Gaussian state, is a bivariate Gaussian distribution which is positive over all phase space. This example is hence a direct verification of Hudson's theorem.

As stated in the main text, the analysis of tunnelling for Gaussian states lies predominantly on the effects. First of all, consider the effect $\mathcal{E}_{\{x | V(x) > E^*\}}$. By the condition $V(x) > E^*$, the region corresponding to each $E^*$ can be found by
\begin{align}
\frac{1}{2}m\omega^2 x^2 &> E^* \nonumber  \\
x^2 &> \frac{2E^*}{m\omega^2} \nonumber \\
|x| &> \sqrt{\frac{2E^*}{m\omega^2}}, 
\end{align}
which implies the function $\mathcal{E}_{\{x | V(x) > E^*\}}$ has the form
\begin{equation}
\label{eq:PSTEXE}
\mathcal{E}_{\{x | V(x) > E^*\}} = \begin{cases}
1 & \text{for }x > \sqrt{2E^*/m\omega^2} \text{ or }x < \minus\sqrt{2E^*/m\omega^2} \\
0 & \text{otherwise},
\end{cases}
\end{equation}
which is simply its classical counterpart. Secondly, the effect $\mathcal{E}_{E > E^*}$ is simply 
\begin{equation}
\label{eq:PSTEEE}
\mathcal{E}_{E > E^*} (x,p) = 2 \pi \hbar \sum_{n=n^*}^\infty W_n(x,p),
\end{equation}
where $W_n(x,p)$ is the Wigner function for $n$-th energy eigenstate of quantum harmonic oscillator, and $n^*$ is the minimum quantum number that corresponds to an energy eigenstate with energy greater than $E^*$, or 
\begin{equation}
\label{eq:PSTN}
n^* = \left\lceil \frac{E^*}{\hbar \omega} - \frac{1}{2} \right\rceil.
\end{equation}

For analysis of tunnelling of ground state, the most relevant energy $E^*$ is the ground state energy at $E^* = \hbar \omega/2$. Therefore, consider the effects $\mathcal{E}_{\{x | V(x) > \hbar \omega / 2\}}$ and $\mathcal{E}_{E > \hbar \omega / 2}$, which can be calculated to be
\begin{align}
\mathcal{E}_{\{x | V(x) > \hbar \omega / 2\}}(x,p) &= \begin{cases}
1 &\text{for }x > \sqrt{\hbar/m\omega} \text{ or }x < \minus\sqrt{\hbar/m\omega} \\
0 &\text{otherwise;}
\end{cases} \\
\nonumber \\
\mathcal{E}_{E > \hbar\omega/2}(x, p) &= 1-2\pi\hbar W_0(x,p) \nonumber \\
&= 1 - 2 \exp\left(\minus\frac{m\omega x^2}{\hbar}\right)\exp\left(\minus\frac{p^2}{m\omega\hbar}\right).
\end{align}

An interesting point to note here is that despite the effect corresponding to the probability $E > \hbar \omega / 2$ is simply the difference between identity $\mathbb{1}$, and the rescaled Wigner function of ground state of quantum harmonic oscillator $W_0(x,p)$, which is a positive function, the function as the difference between the two effects still contains negativities. Despite the two fundamental components to the effect are positive function and can sometimes interpreted as classical effects and distributions, ultimately the combination of the two leads to non-classical behaviours. 

\begin{figure}[!htb]
	\centering
	\includegraphics[scale=0.375]{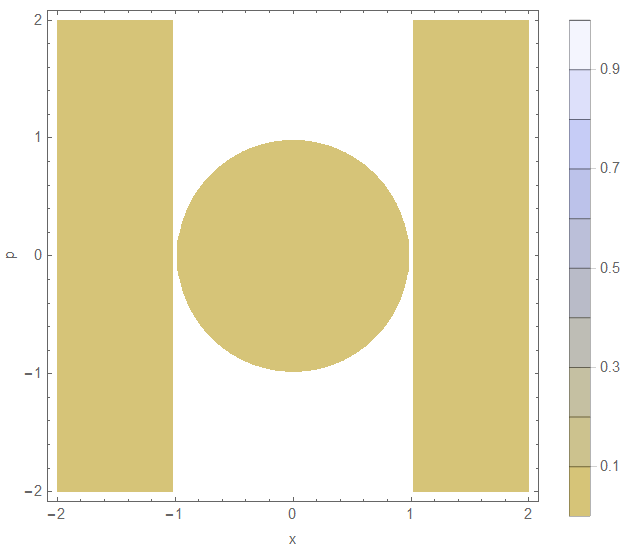}
	\includegraphics[scale=0.375]{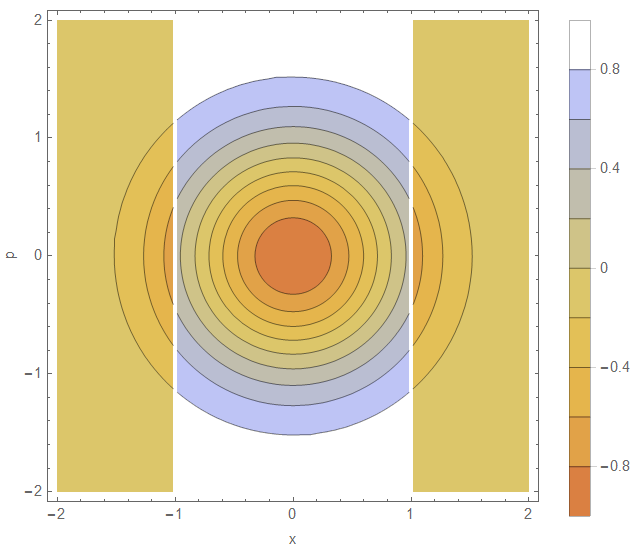}
	\caption{Contour plots of the function $[\mathcal{E}_{E > \hbar\omega/2} - \mathcal{E}_{\{x | V(x) > \hbar \omega / 2\}}](x,p)$ for the classical case (left) and quantum case (right).  \label{fig:QTDE}}
\end{figure}

Moving on with the analysis, the difference between the two effects is
\begin{align}
&[\mathcal{E}_{E > \hbar\omega/2} - \mathcal{E}_{\{x | V(x) > \hbar \omega / 2\}}](x,p) \nonumber \\
=& \begin{cases}
\minus 2 \exp\left(\minus\frac{m\omega x^2}{\hbar}\right)\exp\left(\minus\frac{p^2}{m\omega\hbar}\right) & \text{for }x > \sqrt{\hbar/m\omega} \text{ or }x < \minus\sqrt{\hbar/m\omega} \nonumber \\
1 - 2 \exp\left(\minus\frac{m\omega x^2}{\hbar}\right)\exp\left(\minus\frac{p^2}{m\omega\hbar}\right) & \text{otherwise},
\end{cases} 
\end{align}
which clearly shows that such function contains negativities. By \text{Cor. \ref{cor:PTG}}, the ground state of a quantum harmonic oscillator can indeed tunnel. As a comparison and an example to the discussion in the main text regarding quantum and classical cases of tunnelling, the classical and quantum version of the tunnelling rate operator is shown in \text{Fig. \ref{fig:QTDE}}, which clearly demonstrates that only the quantum case of the function contains negativities. 

Although it is shown by \text{Cor. \ref{cor:PTG}} that the ground state can tunnel, to demonstrate that it is indeed tunnelling is to consider the integral
\begin{align*}
&\iint [\mathcal{E}_{E > \hbar\omega/2} - \mathcal{E}_{\{x | V(x) > \hbar \omega / 2\}}](x,p) W_0(x,p) \,dx \,dp \\
=&\,\minus \iint \mathcal{E}_{\{x | V(x) > \hbar \omega / 2\}}(x,p) W_0(x,p)\,dx\,dp \\
=&\,\minus \iint_{-\infty}^{-\sqrt{\hbar/m\omega}} W_0(x,p) \,dx\,dp  - \iint_{\sqrt{\hbar/m\omega}}^\infty W_0(x,p)\,dx\,dp,
\end{align*}
since $\iint \mathcal{E}_{E > \hbar\omega/2}(x,p) W_0(x,p)\,dx\,dp = 0$. By Hudson's theorem, $W_0(x,p)$ is positive, and therefore the last line of the derivation is negative. By \text{Thm. \ref{thm:PSTNS}}, the ground state of quantum harmonic oscillator is indeed a tunnelling state. 

\section{Orthogonality of Wigner Functions of Eigenstates}
The Weyl transformation between Hermitian operator and real function in phase space provides a method of generating sets of orthogonal functions in phase space. A Hermitian density operator can generally be expressed as $\hat{\rho} = \sum_{i,j} \lambda_{i,j} \ket{\omega_i} \bra{\omega_j}$, where $\ket{\omega_i}$ are eigenbases of another Hermitian operator. Consider, then, the Weyl transformation of the operator $\ket{\omega_i}\bra{\omega_j}$ as $F_{i,j}$, such that for a Wigner function $W(x,p)$,
\begin{equation}
\hat{\rho} = \sum_{i,j} \lambda_{i,j} \ket{\omega_i}\bra{\omega_j} \rightarrow W(x,p) = \sum_{i,j} \lambda_{i,j} F_{i,j}(x,p),
\end{equation}
where the set of functions $F_{i,j}$ are orthogonal,
\begin{align}
&\,\Tr\left(\ket{\omega_{i_1}}\braket{\omega_{j_1}|\omega_{j_2}}\bra{\omega_{i_2}}\right) \nonumber \\
=&\, 2\pi\hbar\iint F_{i_1,j_1}(x,p)F^*_{i_2,j_2}(x,p)\,dx\,dp = \delta_{i_1,j_1} \delta_{i_2,j_2},
\end{align}
and completeness of the operators $\ket{\omega_i}\bra{\omega_j}$ in Hermitian matrices is mapped to the completeness of the corresponding functions $F_{i,j}(x,p)$ in real functions. Therefore, the coefficients $\lambda_{i,j}$ can be calculated by
\begin{equation}
\lambda_{i,j} = \Tr(\hat{\rho}\ket{\omega_j}\bra{\omega_i}) = 2\pi\hbar\iint W(x,p) F^*_{i,j}(x,p)\,dx\,dp.
\end{equation}
An example of such eigenbasis decomposition is momentum eigenbasis $\ket{p_i}\bra{p_j}$, such that
\begin{align}
2\pi\hbar F_{i, j} (x,p) &= \frac{2\pi\hbar}{\pi\hbar}\int e^{2ipy/\hbar} \braket{x-y|p_i}\braket{p_j|x+y}\,dy \nonumber \\
&= \frac{2}{2\pi\hbar} \int e^{2ipy/\hbar} e^{ip_i(x-y)/\hbar} e^{\minus ip_j(x+y)/\hbar}\,dy \nonumber \\
&= \frac{1}{2\pi} \delta[p - \frac{1}{2}(p_i + p_j)]e^{i(p_i-p_j)x/\hbar},
\end{align}
which is simply the Fourier transform in the position coordinate, with $p_i + p_j$ and $p_i - p_j$ as two independent parameters. 

\end{document}